\documentclass[12pt]{article}

\usepackage{latexsym,amsmath,amscd,amssymb,graphics,float}
\usepackage{enumerate}

\usepackage{graphicx,lscape}

\usepackage[colorlinks]{hyperref}
\usepackage{url}

\begin{document}

\newenvironment{proof}[1][Proof]{\textbf{#1.} }{\ \rule{0.5em}{0.5em}}

\newtheorem{theorem}{Theorem}[section]
\newtheorem{definition}[theorem]{Definition}
\newtheorem{lemma}[theorem]{Lemma}
\newtheorem{remark}[theorem]{Remark}
\newtheorem{proposition}[theorem]{Proposition}
\newtheorem{corollary}[theorem]{Corollary}
\newtheorem{example}[theorem]{Example}

\numberwithin{equation}{section}
\newcommand{\ep}{\varepsilon}
\newcommand{\R}{{\mathbb  R}}
\newcommand\C{{\mathbb  C}}
\newcommand\Q{{\mathbb Q}}
\newcommand\Z{{\mathbb Z}}
\newcommand{\N}{{\mathbb N}}

\newcommand{\bfi}{\bfseries\itshape}

\newsavebox{\savepar}
\newenvironment{boxit}{\begin{lrbox}{\savepar}
\begin{minipage}[b]{15.5cm}}{\end{minipage}\end{lrbox}
\fbox{\usebox{\savepar}}}

\title{{\bf On a Hamiltonian version of a 3D Lotka-Volterra system}}
\author{R\u{a}zvan M. Tudoran and Anania G\^\i rban}

\date{}
\maketitle \makeatother

\begin{abstract}
In this paper we present some relevant dynamical properties
of a 3D Lotka-Volterra system from the Poisson dynamics point of
view.
\end{abstract}

\medskip

\textbf{AMS 2000}: 70H05; 37J25; 37J35.

\textbf{Keywords}: Hamiltonian dynamics, Lotka-Volterra
system, stability of equilibria, Poincar\'e compactification, energy-Casimir mapping.

\section{Introduction}
\label{section:one}

The Lotka-Volterra system has been widely investigated in the last
years. This system, studied by May and Leonard \cite{may}, models the evolution of competition between three species. Among
the studied topics related with the Lotka-Volterra system, we
recall a few of them together with a partial list of references, namely: integrals and invariant manifolds (\cite{leach1}, \cite{llibre}), stability (\cite{may1}, \cite{leach1}), analytic behavior \cite{leach}, nonlinear analysis \cite{may}, and many others.

In this paper we consider a special case of the Lotka-Volterra system, recently introduced in \cite{llibre}. We write the system as a Hamiltonian system of Poisson type in order to analyze the system from
the Poisson dynamics point of view. More exactly, in the second
section of this paper, we prepare the framework of our study
by writing the Lotka-Volterra system as a Hamilton-Poisson system, and
also find a $SL(2,\mathbb{R})$ parameterized family of
Hamilton-Poisson realizations. As consequence of the Hamiltonian setting we obtain two new first integrals of the Lotka-Volterra system that generates the first integrals of this system found in \cite{llibre}. In the third section of the paper
we determine the equilibria of the Lotka-Volterra system and then
analyze their Lyapunov stability. The fourth section is
dedicated to the study of the Poincar\'e compactification of the Lotka-Volterra system. More exactly, we integrate explicitly the Poincar\'e compactification of the Lotka-Volterra system. In the fifth section of the article we present some
convexity properties of the image of the energy-Casimir mapping
and define some naturally associated semialgebraic splittings of
the image. More precisely, we discuss the relation between
the image through the energy-Casimir mapping of the families of equilibria of the Lotka-Volterra
system and the canonical Whitney stratifications of the
semialgebraic splittings of the image of the energy-Casimir
mapping. In the sixth part of the paper we give a
topological classification of the fibers of the energy-Casimir
mapping, classification that follows naturally from the
stratifications introduced in the above section. Note that in our
approach we consider fibers over the regular and also over
the singular values of the energy-Casimir mapping. In the last part of the article we give two Lax
formulations of the system. For details on Poisson geometry and
Hamiltonian dynamics see e.g. \cite{marsden}, \cite{marsdenratiu},
\cite{cushman}, \cite{ratiu}.
\medskip

\section{Hamilton-Poisson realizations of a 3D Lotka-Volterra system}
The Lotka-Volterra system we consider for our study, is governed by the equations:
\begin{equation}\label{sys}
\left\{ \begin{array}{l}
 \dot x = -x(x-y-z) \\
 \dot y = -y(-x+y-z) \\
 \dot z = -z(-x-y+z). \\
 \end{array} \right.
\end{equation}
Note that the above system is the Lotka-Volterra system studied in \cite{llibre} in the case $a=b=-1$. 

In \cite{llibre} it is shown that this system admits the following polynomial conservation laws:
$$
f(x,y,z)=xyz(x-y)(x-z)(y-z),
$$
$$
g(x,y,z)=x^2 y^2 -x^2 yz -xy^2 z+x^2 z^2 -xy z^2 +y^2 z^2.
$$

Using Hamiltonian setting of the problem, we provide two degree-two polynomial conservation laws of the system \eqref{sys} which generates $f$ and $g$. These conservation laws will be represented by the Hamiltonian and respectively a Casimir function of the Poisson configuration manifold of the system \eqref{sys}.  

As the purpose of this paper is to study the above system from the
Poisson dynamics point of view, the first step in this approach is
to give a Hamilton-Poisson realization of the system.
\begin{theorem}\label{t22}
The dynamics \eqref{sys} has the following Hamilton-Poisson
realization:
\begin{equation}\label{hc}
(\R^3,\Pi_{C},H)
\end{equation}
where,
$$\Pi_{C}(x,y,z)=\left[ {\begin{array}{*{20}c}
   0 & y & {x-z}  \\
   -y & 0 & {-y}  \\
   {-x+z} & {y} & 0  \\
\end{array}} \right]$$
is the Poisson structure generated by the smooth function
$C(x,y,z):=-xy+yz$, and the Hamiltonian
$H\in C^\infty(\R^3,\R)$ is given by
$H(x,y,z):=xy-xz$.

Note that, by Poisson structure generated by the smooth function
$C$, we mean the Poisson structure generated by the Poisson
bracket $\{f,g\}:=\nabla C\cdot(\nabla f\times\nabla g)$, for any
smooth functions $f,g\in C^\infty(\R^3,\R)$.
\end{theorem}
\noindent \begin{proof} Indeed, we have successively:
$$\Pi_{C}(x,y,z)\cdot\nabla H(x,y,z)=
\left[ {\begin{array}{*{20}c}
   0 & {y} & {x-z}  \\
   {-y} & 0 & {-y }  \\
   {-x+z} & {y} & 0  \\
\end{array}} \right]\cdot
\left[ {\begin{array}{*{20}c}
   {y-z} \\
   {x}  \\
   {-x} \\
\end{array}} \right]=
\left[ {\begin{array}{*{20}l}
   {-x(x-y-z)} \\
   { -y(-x+y-z)}  \\
   {-z(-x-y+z)} \\
\end{array}} \right]=
\left[ {\begin{array}{*{20}r}
   {\dot x} \\
   {\dot y}  \\
   {\dot z} \\
\end{array}} \right],
$$
as required.
\end{proof}
\medskip
\begin{remark}
Since the signature of the quadratic form generated by $C(x,y,z)=-xy+yz$ is $(-1,0,+1)$, the triple $(\mathbb{R}^{3},\Pi_{C},H)$ it is isomorphic with a Lie-Poisson realization of the Lotka-Volterra system \eqref{sys} on the dual of the semidirect product between the Lie algebra $\mathfrak{so}(1,1)$ and $\mathbb{R}^2$.
\end{remark}

\begin{remark}
By definition we have that the center of the Poisson algebra
$C^\infty(\mathbb{R}^{3},\R)$ is generated by the Casimir
invariant $C(x,y,z)=-xy+yz$.
\end{remark}

\begin{remark}
The conservation laws $f$ and $g$ found in \cite{llibre}, can be written in terms of the Casimir $C$ and respectively the Hamiltonian $H$ as follows:
$$
f=CH(C+H),
$$
$$
g=\dfrac{1}{2}\left[C^2+H^2+(C+H)^2\right].
$$
\end{remark}

Next proposition gives others Hamilton-Poisson realizations of the Lotka-Volterra system \eqref{sys}.
\begin{proposition}
The dynamics \eqref{sys} admits a family of Hamilton-Poisson
realizations parameterized by the group $SL(2,\R)$. More exactly,
$(\R^3,\{\cdot,\cdot\}_{a,b},H_{c,d})$ is a
Hamilton-Poisson realization of the dynamics \eqref{sys} where the
bracket $\{\cdot,\cdot\}_{a,b}$ is defined by
$$\{f,g\}_{a,b}:=\nabla C_{a,b}\cdot(\nabla f\times\nabla g),$$
for any $f,g\in C^\infty(\R^3,\R)$, and the functions
$C_{a,b}$ and $H_{c,d}$ are given by:
$$C_{a,b}(x,y,z)=(-a+b)xy+ayz-bxz
,$$
$$H_{c,d}(x,y,z)=(-c+d)xy+cyz-dxz
,$$ respectively, the
matrix of coefficients $a,b,c,d$ is $\left[ {\begin{array}{*{20}c}
   a & b  \\
   c & d  \\
\end{array}} \right]\in SL(2,\R).$
\end{proposition}
\noindent \begin{proof} The conclusion follows directly taking
into account that the matrix formulation of the Poisson bracket
$\{\cdot,\cdot\}_{a,b}$ is given in coordinates by:
$$\Pi_{a,b}(x,y,z)=
\left[ {\begin{array}{*{20}c}
   0 & {-bx+ay} & {(a-b)x-az} \\
   {bx-ay} & 0 & {(-a+b)y-bz}  \\
   {(-a+b)x+az} & {(a-b)y+bz} & 0
\end{array}} \right].$$
\end{proof}
\medskip

\section{Stability of equilibria}
In this short section we analyze the stability properties of the
equilibrium states of the Lotka-Volterra system \eqref{sys}.
\begin{remark}
The equilibrium states of the system \eqref{sys} are given as the union
of the following three families:
$${\mathcal E}_{1}:=\{(0,M,M):M\in\R\},$$
$${\mathcal E}_{2}:=\{(M,0,M):M\in\R\},$$
$${\mathcal E}_{3}:=\{(M,M,0):M\in\R\}.$$
\end{remark}

Figure \ref{figech} presents the above defined families of equilibrium states of the Lotka-Volterra system.

\begin{figure}[H]
\vspace*{25pt}
\centering
\begin{tabular}{|c|}
\hline
\scalebox{0.9}{%
\includegraphics*{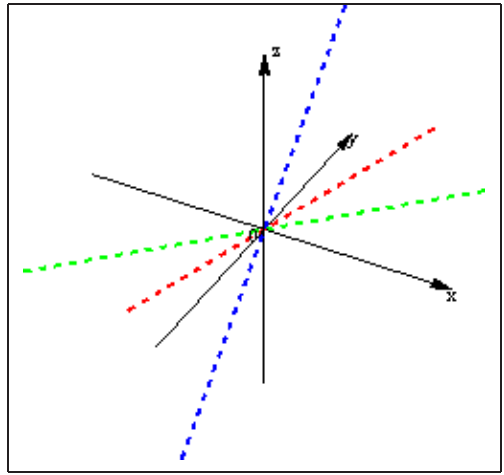}}\\
\hline
\end{tabular}
\caption{\it Equilibria of the Lotka-Volterra system \eqref{sys}.}\label{figech}
\end{figure}

In the following theorem we describe the stability
properties of the equilibrium states of the system \eqref{sys}.
\begin{theorem}\label{t33}
All the equilibrium states of the Lotka-Volterra system \eqref{sys} are unstable.
\end{theorem}
\noindent \begin{proof} 
The conclusion follows from the fact that the characteristic polynomial associated with the linear
part of the system evaluated at an arbitrary equilibrium state, is the same for any of the families $\mathcal {E}_{1},\mathcal {E}_{2},\mathcal {E}_{3}$, and is given by:
$$p(\lambda)=(2M-\lambda)\lambda(2M+\lambda).$$
For $M=0$, we get the origin $(0,0,0)$ which is also unstable since in any arbitrary small open neighborhood around, there exists unstable equilibrium states.
\end{proof}
\medskip

\section{The behavior on the sphere at infinity}

In this section we integrate explicitly the Poincar\'e compactification of the Lotka-Volterra system, and consequently the Lotka-Volterra system \eqref{sys} (on the sphere) at infinity. Recall that using the Poincar\'e compactification of $\R^3$, the infinity of $\R^3$ is represented by the sphere $\mathbb{S}^2$ - the equator of the unit sphere $\mathbb{S}^3$ in $\R^4$. For details regarding the Poincar\'e compactification of polynomial vector fields in $\R^3$ see \cite{llibre1}.

Fixing the notations in accordance with the results stated in \cite{llibre1} we write the Lotka-Volterra system \eqref{sys} as 
\begin{equation*}
\left\{ \begin{array}{l}
 \dot x = P^1(x,y,z) \\
 \dot y = P^2(x,y,z) \\
 \dot z = P^3(x,y,z), \\
 \end{array} \right.
\end{equation*}
with $P^1(x,y,z)=-x(x-y-z)$, $P^2(x,y,z)=-y(-x+y-z)$, $P^3(x,y,z)=-z(-x-y+z)$.

Let us now study the Poincar\'e compactification of the Lotka-Volterra system in the local charts $U_i$ and $V_i$, $i\in\{1,2,3\}$, of the manifold $\mathbb{S}^3$. 

The Poincar\'e compactification ($p(X)$ in the notations from \cite{llibre1}) of the Lotka-Volterra system \eqref{sys} is the same for each of the local charts $U_1$, $U_2$ and respectively $U_3$ and is given in the corresponding local coordinates by

\begin{equation}\label{compsys}
\left\{ \begin{array}{l}
 \dot z_1 = 2z_1 (1-z_1) \\
 \dot z_2 = 2z_2 (1-z_2) \\
 \dot z_3 = -z_3 (z_1 +z_2 +1). \\
 \end{array} \right.
\end{equation}

Regarding the Poincar\'e compactification of the Lotka-Volterra system in the local charts $V_1$, $V_2$ and respectively $V_3$, as a property of the compactification procedure, the compactified vector field $p(X)$ in the local chart $V_i$ coincides with the vector field $p(X)$ in $U_i$ multiplied by the factor $-1$, for each $i\in\{1,2,3\}$. Hence, for each $i\in\{1,2,3\}$, the flow of the system \eqref{compsys} on the local chart $V_i$ is the same as the flow on the local chart $U_i$ reversing the time.

The system \eqref{compsys} is integrable with the solution given by

\begin{equation*}
\left\{ \begin{array}{l}
 z_1(t) = \dfrac{e^{2t}}{e^{2t}+e^{k_1}} \\
 z_2(t) = \dfrac{e^{2t}}{e^{2t}+e^{k_2}}\\
 z_3(t) = \dfrac{e^{t}k_3}{\sqrt{e^{2t}+e^{k_1}}\sqrt{e^{2t}+e^{k_2}}}, \\
 \end{array} \right.
\end{equation*}
where $k_1, k_2, k_3$ are arbitrary real constants.

To analyze the Lotka-Volterra system on the sphere $\mathbb{S}^2$ at infinity, note that the points on the sphere at infinity are characterized by $z_3=0$. As the plane $z_1 z_2$ is invariant under the flow of the system \eqref{compsys}, the compactified Lotka-Volterra system on the local charts $U_i$ ($i\in\{1,2,3\}$) on the infinity sphere reduces to

\begin{equation}\label{infsys}
\left\{ \begin{array}{l}
 \dot z_1 = 2z_1 (1-z_1) \\
 \dot z_2 = 2z_2 (1-z_2).\\
 \end{array} \right.
\end{equation}

The system \eqref{infsys} is integrable and the solution is given by 
\begin{equation*}
\left\{ \begin{array}{l}
 z_1(t) = \dfrac{e^{2t}}{e^{2t}+e^{k_1}} \\
 z_2(t) = \dfrac{e^{2t}}{e^{2t}+e^{k_2}},\\
 \end{array} \right.
\end{equation*}
where $k_1, k_2$ are arbitrary real constants.

Consequently, the phase portrait on the local charts $U_i$ ($i\in\{1,2,3\}$) on the infinity sphere is given in Figure \ref{sphe}.

\begin{figure}[H]
\vspace*{25pt}
\centering
\scalebox{0.60}{%
\includegraphics*{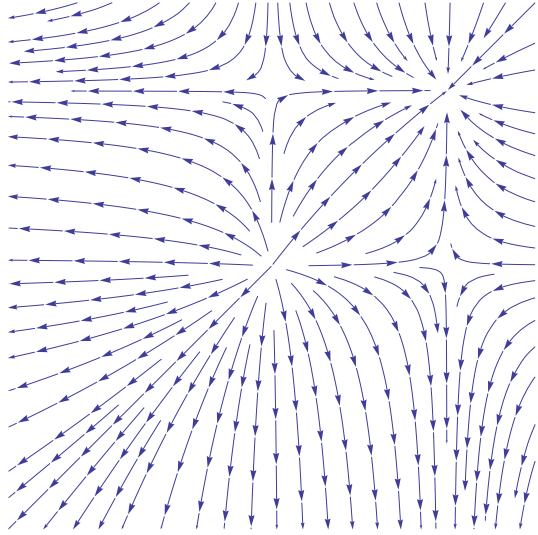}}
\caption{\it Phase portrait of the system \eqref{infsys}.} \label{sphe}
\end{figure}

\section{The image of the energy-Casimir mapping}

The aim of this section is to study the image of the energy-Casimir mapping $Im(\mathcal EC)$, associated with the Hamilton-Poisson realization \eqref{hc} of the Lotka-Volterra system \eqref{sys}.

We consider convexity properties of the image of $\mathcal EC$, as well as a semialgebraic splitting of the image that agree with the topology of the symplectic leaves of the Poisson manifold $(\R^3,\Pi_C)$. Recall that by a semialgebraic splitting, we mean a splitting consisting of semialgebraic manifolds, namely manifolds that are described in coordinates by a set of polynomial inequalities and equalities. For details on semialgebraic manifolds and their geometry see e.g. \cite{pflaum}.

All these will be used later on, in order to obtain a topological classification of the orbits of \eqref{sys}.

Recall first that the energy-Casimir mapping, ${\mathcal EC}\in C^\infty(\R^3,\R^2)$ is given by:
$${\mathcal EC}(x,y,z)=(H(x,y,z), C(x,y,z)),\ (x,y,z)\in\R^3$$
where $H,C\in C^\infty(\R^3,\R)$ are the Hamiltonian of the system \eqref{sys}, and respectively the Casimir of the Poisson manifold $(\R^3,\Pi_C)$, both of them as considered in Theorem \ref{t22}.

Next proposition explicitly gives the semialgebraic splitting of the image of the energy-Casimir map ${\mathcal EC}$.

\begin{proposition}\label{p51}
 The image of the energy-Casimir map - $\R^2$ - admits the following splitting:

$Im(\mathcal EC)=S^{c>0}\cup S^{c=0}\cup S^{c<0}$, where the subsets
$S^{c>0}$, $S^{c=0}$, $S^{c<0}\subset\R^2$ are splitting further on a union of semialgebraic manifolds, as follows:
\begin{align*}
 S^{c>0}&=\{(h,c)\in\R^2:h<0;c>0\}\cup\{(h,c)\in\R^2:h=0;c>0\}\\
&\cup\{(h,c)\in\R^2:h>0;c>0\},\\
 S^{c=0}&=\{(h,c)\in\R^2:h<0;c=0\}\cup\{(h,c)\in\R^2:h=c=0\}\\
&\cup\{(h,c)\in\R^2:h>0;c=0\},\\
 S^{c<0}&=\{(h,c)\in\R^2:c<\min\{-h,0\}\}\cup\{(h,c)\in\R^2:c<0;c=-h\}\\
&\cup\{(h,c)\in\R^2:-h<c<0\}.
\end{align*}
\end{proposition}
\begin{proof}
 The conclusion follows directly by simple algebraic computation using the definition of the energy-Casimir mapping.
\end{proof}
\medskip

\begin{remark}
The superscripts used to denote the sets $S$, are in agreement with the topology of the symplectic leaves of the Poisson manifold $(\R^3,\Pi_C)$, namely:
\begin{itemize}
\item[(i)] For $c\neq 0$,
$${\Gamma}_c=\{(x,y,z)\in\R^3:y(z-x)=c\}$$
is a hyperbolic cylinder.
\item[(ii)] For $c=0$,
$${\Gamma}_c=\{(x,y,z)\in\R^3:y(z-x)=c\}$$
is a union of two intersecting planes.
\end{itemize}
\end{remark}
The connection between the semialgebraic splittings of the image $Im({\mathcal EC})$ given by Proposition \ref{p51}, and the equilibrium states of the Lotka-Volterra system, is given in the following remark.

\begin{remark}\label{strt}
The semialgebraic splitting of the sets $S$ is described in terms of the image of equilibria of the Lotka-Volterra system through the map $\mathcal EC$ as follows:

\begin{itemize}
\item[(i)]
\begin{align*}
\ S^{c>0}&=\Sigma_{1}^{\leftarrow}\cup Im(\left.{\mathcal EC}\right|_{{\mathcal E}_{1}^{\star}})\cup
\Sigma_{1}^{\rightarrow},\ where \ {\mathcal E}_{1}^{\star}={\mathcal E}_{1}\setminus\{(0,0,0)\}, \\
\Sigma_{1}^{\leftarrow}&=\{(h,c)\in\R^2:h<0;c>0\},\ \Sigma_{1}^{\rightarrow}=\{(h,c)\in\R^2:h>0;c>0\}.
\end{align*}
\item[(ii)]
\begin{align*}
\  S^{c=0}&=Im(\left.{\mathcal EC}\right|_{{\mathcal E}_{2}^{\star}})\cup \Sigma_{0}\cap \Sigma_{0}^{\rightarrow},\ where \ {\mathcal E}_{2}^{\star}={\mathcal E}_{2}\setminus\{(0,0,0)\}, \\
\Sigma_{0}&=\{(0,0,0)\},\ \Sigma_{0}^{\rightarrow}=\{(h,c)\in\R^2:h>0;c=0\}.
\end{align*}

\item[(iii)]
\begin{align*}
\  S^{c<0}&=\Sigma_{3}^{\leftarrow}\cup Im(\left.{\mathcal EC}\right|_{{\mathcal E}_{3}^{\star}})\cup
\Sigma_{3}^{\rightarrow},\ where \ {\mathcal E}_{3}^{\star}={\mathcal E}_{3}\setminus\{(0,0,0)\}, \\
\Sigma_{3}^{\leftarrow}&=\{(h,c)\in\R^2:c<\min\{-h,0\}\},\ \Sigma_{3}^{\rightarrow}=\{(h,c)\in\R^2:-h<c<0\}.
\end{align*}
\end{itemize}
\end{remark}

All the stratification results can be gathered as shown in Figure \ref{fig3}. Note that for simplicity we adopted the notation $Im(\left.{\mathcal EC}\right|_{{\mathcal E}})\mathop=\limits^{not}:
\Sigma$.

\begin{figure}[H]
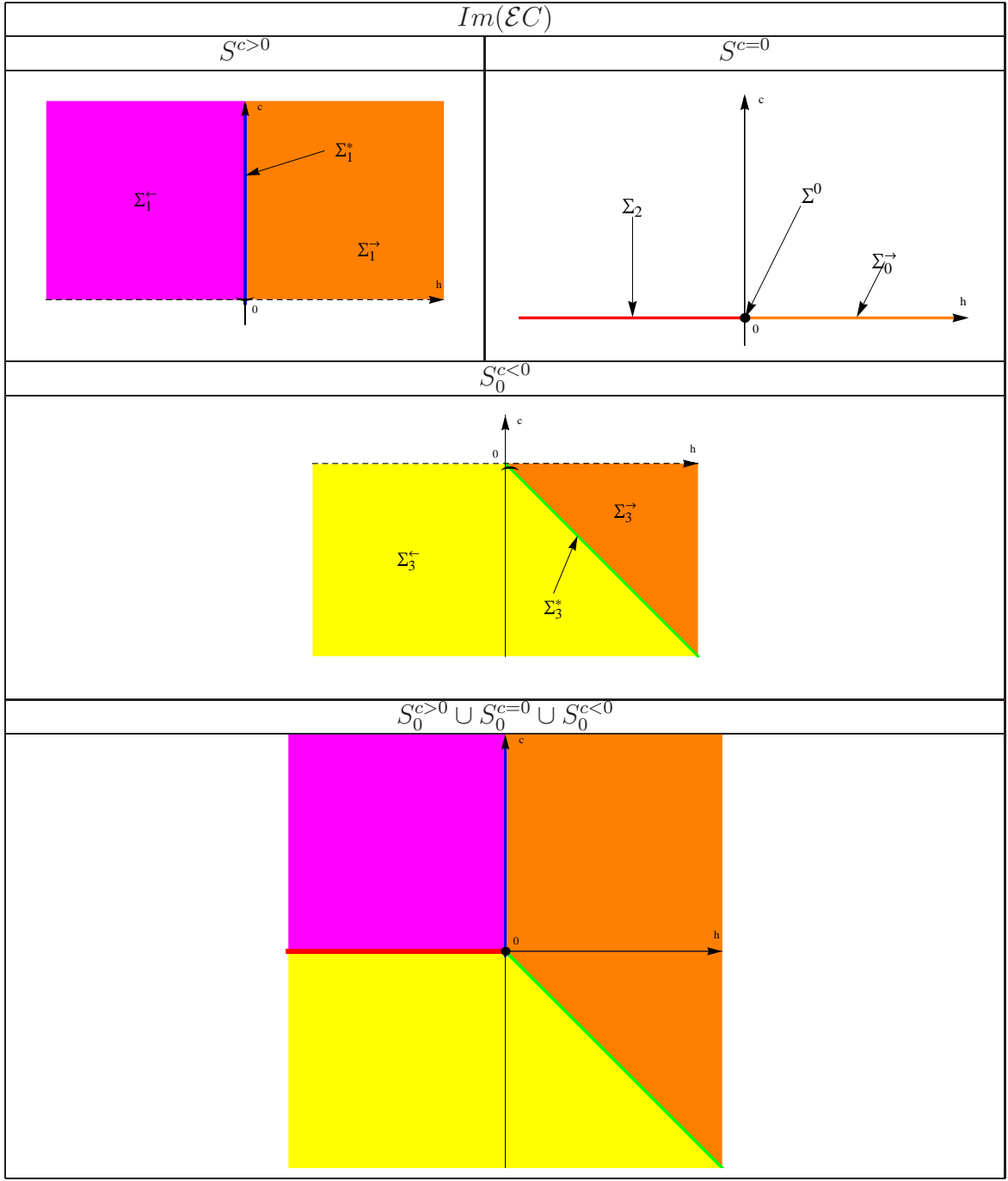

\vspace*{25pt}
\centering
\begin{tabular}{|c|c|}
\hline
\multicolumn{2}{|c|}{$Im({\mathcal EC})$}\\
\hline
$S^{c>0}$ & $S^{c=0}$\\
\hline
\scalebox{0.55}{%
\includegraphics*{im1.%
eps}} &
\scalebox{0.60}{%
\includegraphics*{im2.%
eps}} \\
\hline
\multicolumn{2}{|c|}{$S_0^{c<0}$ }\\
\hline
\multicolumn{2}{|c|}{\scalebox{0.55}{%
\includegraphics*{im3.%
eps}} }\\
\hline
\multicolumn{2}{|c|}{$S_0^{c>0}\cup S_0^{c=0}\cup S_0^{c<0}$}\\
\hline
\multicolumn{2}{|c|}{\scalebox{0.55}{%
\includegraphics*{imag.%
eps}}} \\
\hline
\end{tabular}
\caption{\it Semialgebraic splitting of $Im({\mathcal EC})$.}\label{fig3}
\end{figure}

\begin{remark}
As a convex set, the image of the energy-Casimir map is convexly generated by the images of the equilibrium states of the Lotka-Volterra system \eqref{sys}, namely:
$$Im({\mathcal EC})=\overline{\rm co}\{
Im(\left.{\mathcal EC}\right|_{{\mathcal E}_{1}}),Im(\left.{\mathcal EC}\right|_{{\mathcal E}_{2}}),Im(\left.{\mathcal EC}\right|_{{\mathcal E}_{3}})\}.$$
\end{remark}

\section{The topology of the fibers of the energy-Casimir mapping}

In this section we describe the topology of the fibers of $\mathcal EC$, considering for our study fibers over regular values of $\mathcal EC$ as well as fibers over the singular values. It will remain an open question how these fibers fit all together in a more abstract fashion, such as bundle structures in the symplectic Arnold-Liouville integrable regular case.

\begin{proposition}
According to the stratifications from the previous section, the topology of the fibers of $\mathcal EC$ can be described as in Tables \ref{tab1}, \ref{tab2}, \ref{tab3}:
\vskip1cm
\begin{table}[H]
\centering
\begin{tabular}{|c|c|c|c|}
\hline
\multicolumn{4}{|c|}{ $S^{c>0}$}\\
\hline
$A\subseteq S^{c>0}$ & $\Sigma_1^{\leftarrow}$ & $\Sigma_1^*$ & $\Sigma_1^{\rightarrow}$\\
\hline
${\cal F}_{(h,c)}\subseteq\R^3$ & \scalebox{0.2}{%
\includegraphics*{magenta.%
eps}}& \scalebox{0.2}{%
\includegraphics*{blue.%
eps}} & \scalebox{0.2}{%
\includegraphics*{orange.%
eps}}\\ \cline{2-4}
$(h,c)\in A$ & $\mathop{\coprod}\limits_{i = 1}^4(\R\times\{i\})$ & $\mathop{\coprod}\limits_{i = 1}^8(\R\times\{i\})\mathop{\coprod}\{pt\}\mathop{\coprod}\{pt'\}$ & $\mathop{\coprod}\limits_{i = 1}^4(\R\times\{i\})$\\
\hline
Dynamical & union of & union of 8 & union of \\
description & 4 orbits & orbits and two& 4 orbits \\
 &  & equilibrium points &  \\
\hline
\end{tabular}
\caption{\it Fibers classification corresponding to $S^{c>0}$.}\label{tab1}
\end{table}

\vskip1cm
\begin{table}[H]
\centering
\begin{tabular}{|c|c|c|c|}
\hline
\multicolumn{4}{|c|}{ $S^{c=0}$}\\
\hline
$A\subseteq S^{c=0}$ & $\Sigma_2$ & $\Sigma^0$ & $\Sigma_0^{\rightarrow}$\\
\hline
${\cal F}_{(h,c)}\subseteq\R^3$ & \scalebox{0.2}{%
\includegraphics*{red.%
eps}}& \scalebox{0.2}{%
\includegraphics*{black.%
eps}} & \scalebox{0.2}{%
\includegraphics*{orange.%
eps}}\\ \cline{2-4}
$(h,c)\in A$ & $\coprod\limits_{i = 1}^8(\R\times\{i\})\coprod\{pt\}\coprod\{pt'\}$ & $\coprod\limits_{i = 1}^8(\R\times\{i\})\coprod\{pt\}$ & $\coprod\limits_{i = 1}^4(\R\times\{i\})$\\
\hline
Dynamical & union of 8 & union of 8 & union of \\
description & orbits and two & orbits and one& 4 orbits \\
 & equilibrium points & equilibrium point &  \\
\hline
\end{tabular}
\caption{\it Fibers classification corresponding to $S^{c=0}$.}\label{tab2}
\end{table}
\vskip1cm
\begin{table}[H]
\centering
\begin{tabular}{|c|c|c|c|}
\hline
\multicolumn{4}{|c|}{ $S^{c<0}$}\\
\hline
$A\subseteq S^{c<0}$ & $\Sigma_3^{\leftarrow}$ & $\Sigma_3^*$ & $\Sigma_3^{\rightarrow}$\\
\hline
${\cal F}_{(h,c)}\subseteq\R^3$ & \scalebox{0.2}{%
\includegraphics*{yellow.%
eps}}& \scalebox{0.2}{%
\includegraphics*{green.%
eps}} & \scalebox{0.2}{%
\includegraphics*{orange.%
eps}}\\ \cline{2-4}
$(h,c)\in A$ & $\mathop{\coprod}\limits_{i = 1}^4(\R\times\{i\})$ & $\mathop{\coprod}\limits_{i = 1}^8(\R\times\{i\})\mathop{\coprod}\{pt\}\mathop{\coprod}\{pt'\}$ & $\mathop{\coprod}\limits_{i = 1}^4(\R\times\{i\})$\\
\hline
Dynamical & union of & union of 8 & union of \\
description & 4 orbits & orbits and two& 4 orbits \\
 &  & equilibrium points &  \\
\hline
\end{tabular}
\caption{\it Fibers classification corresponding to $S^{c<0}$.}\label{tab3}
\end{table}

\end{proposition}
\noindent \begin{proof}
The conclusion follows by simple computations according to the topology of the solution set of the system:
$$\left\{ \begin{array}{l}
 H(x,y,z) = h \\
 C(x,y,z) = c \\
 \end{array} \right.$$
 where $(h,c)$ belongs to the semialgebraic manifolds introduced in the above section.
\end{proof}
\medskip

A presentation that puts together the topological classification of the fibers of $\mathcal EC$ and the topological classification of the symplectic leaves of the Poisson manifold $(\R^3,\Pi_C)$, is given in Figure \ref{fig6}.

\begin{figure}[H]
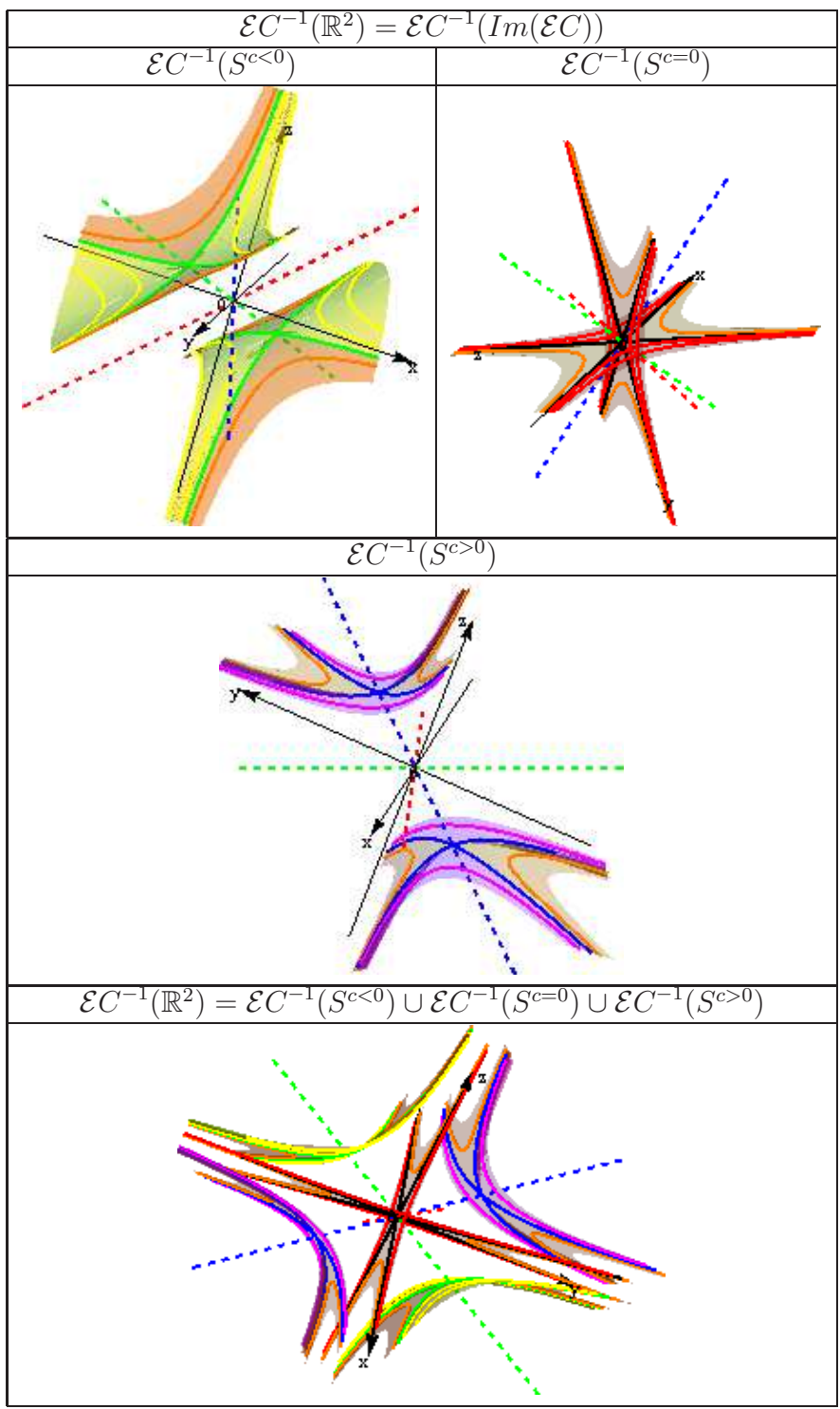

\centering
\vspace*{20pt}
\index{\footnote{}}
\begin{tabular}{|c|c|}
\hline
\multicolumn{2}{|c|}{${\mathcal EC}^{-1}(\R^2)={\mathcal EC}^{-1}(Im({\mathcal EC}))$}\\
\hline
${\mathcal EC}^{-1}(S^{c<0})$ & ${\mathcal EC}^{-1}(S^{c=0})$ \\
\hline
\scalebox{0.9}{%
\includegraphics*{sigmam15.%
eps}}
& \scalebox{0.9}{%
\includegraphics*{sigma0.%
eps}} \\
\hline
\multicolumn{2}{|c|}{${\mathcal EC}^{-1}(S^{c>0})$}  \\
\hline
\multicolumn{2}{|c|}{\scalebox{0.9}{%
\includegraphics*{sigma15.%
eps}}} \\
\hline
\multicolumn{2}{|c|}{${\mathcal EC}^{-1}(\R^2)={\mathcal EC}^{-1}(S^{c<0})\cup{\mathcal EC}^{-1}(S^{c=0})\cup{\mathcal EC}^{-1}(S^{c>0})$}  \\
\hline
\multicolumn{2}{|c|}{\scalebox{0.9}{%
\includegraphics*{sigma.%
eps}}} \\
\hline
\end{tabular}
\caption{\it Phase portrait splitting.}\label{fig6}
\end{figure}

\section{Lax Formulation}

In this section we present a Lax formulation of the Lotka-Volterra system \eqref{sys}.

Let us first note that as the system \eqref{sys} restricted to a
regular symplectic leaf, give rise to a sypmlectic Hamiltonian
system that is completely integrable in the sense of Liouville and
consequently it has a Lax formulation.

Is a natural question to ask if the unrestricted system admit a
Lax formulation. The answer is positive and is given by the
following proposition:

\begin{proposition}
The Lotka-Volterra system \eqref{sys} can be written in the Lax
form $\dot L=[L,B]$, where the matrices $L$ and respectively $B$
are given by:
$$
L= \left[{\begin{array}{*{20}c}
   0 & {x - y} & {z}  \\
   { - x + y} & 0 & {i(x+y-z)}  \\
   { - z} & { -i(x+y-z)} & 0  \\
\end{array}} \right],
\ B= \left[{\begin{array}{*{20}c}
   0 & {iz} & {i(x-y)}  \\
   {-iz} & 0 & 0  \\
   { -i(x-y)} & 0 & 0  \\
\end{array}} \right].$$
\end{proposition}

\bigskip
\bigskip

\noindent {\sc R.M. Tudoran}\\
The West University of Timi\c soara\\
Faculty of Mathematics and C.S., Department of Mathematics,\\
B-dl. Vasile Parvan, No. 4,\\
300223-Timi\c soara, Romania.\\
E-mail: {\sf tudoran@math.uvt.ro}\\
Supported by CNCSIS-UEFISCDI, project number PN II-IDEI code 1081/2008 No. 550/2009.
\medskip

\noindent {\sc A. G\^\i rban}\\
"Politehnica" University of Timi\c soara\\
Department of Mathematics, \\
Pia\c ta Victoriei nr. 2,\\
300006-Timi\c soara, Rom\^ania.\\
E-mail: {\sf anania.girban@gmail.com}
\medskip

\end{document}